\documentclass[envcountsect]{llncs}
\usepackage{pgf,tikz}
\usetikzlibrary{trees}
\usepackage{latexsym,amssymb}
\usepackage{graphics}
\usepackage{enumerate}
\usepackage{verbatim}
\usepackage{float}

\usepackage{a4wide}

\title{Characterizing Successful Formulas: the Multi-agent Case}
\author{Sanchit Saraf \and Sumit Sourabh}
\institute{Institute for Logic, Language and Computation, Universiteit van Amsterdam,\\
P.O. Box 94242, 1090 GE Amsterdam, The Netherlands.\\
\email{\{s.saraf, s.sourabh\}@uva.nl}
}
\begin{document}
\maketitle   
\begin{abstract}
Characterization of successful formulas in Public Announcement Logic (PAL) is a well known open problem in Dynamic Epistemic Logic. Recently, Holliday and ICard have given a complete characterization for the single agent case  in \cite{holliday}. However, the problem for the multi-agent case is open. This paper gives a partial solution to the problem, characterizing the subclass of the language consisting of unary operators, and discusses methods to give a complete solution.
\end{abstract}

\newenvironment{replemma}[1]{\noindent {\bf Lemma~#1.}}{\smallskip}

\section{Introduction}
The logic of Public Announcements is the simplest form of $S5$ dynamic epistemic
logic, augmenting standard epistemic logic with public announcement operator. It was formulated and axiomatised without the common knowledge operator by Plaza in \cite{plaza}. The axiomatisation of Public announcement logic (PAL) with the common knowledge operator was given by Baltag, Moss and Solecki \cite{BaltagMS98}. In the same paper, the authors show that PAL is not strongly complete because of infinitary nature of the common knowledge operator. For a detailed account on PAL, one can refer to  \cite{DEL}.

 The notion of a successful update was given by Gebrandy \cite{gebrandy} and van Ditmarsch \cite{ditmarschphd}. The formulas which remain true after being announced are called successful formulas. An interesting open problem in PAL is the syntactic characterization of
successful formulas \cite{benthemld,BenthemOpen,benthemone,gebrandy,DEL,secret,BaltagvDM07}. A classic example of a formula which is not successful is the Moore Sentence $p \land \neg Kp$  \cite{moore}, which can be read as ``$p$ is true but you do not know $p$". The Moore sentences have been analysed by Hintikka in his classical monograph \cite{Hintikka}. Their relevance has been extensively studied by van Ditmarsch and Kooi in their paper \cite{secret}. Successful formulas have important applications in many security protocols.  Together with its practical usefulness, the task of characterizing successful formulas independently presents itself as an interesting mathematical problem.

The aim of this paper is to present a partial solution to the problem for the multi-agent case. Other solutions have been proposed, most notably by Holliday and Icard in their recent paper \cite{holliday}, where they completely solve the problem for the single agent case. In \cite{holliday}, it is also shown that for a single agent, the source of failure is Moorean in nature, which implies that unsuccessful formulas contain at least one binary operator. In contrast, for the multi-agent case, formulas with only unary operators can also be unsuccessful. The simplest
example of such a formula is $K_aL_b p$. The full syntactic characterization of successful formulas in the multi-agent case is a difficult task and not a simple generalization of the single agent case.

 We give a characterization for the successful formulas in the fragment $\mathcal{L}_\mathrm{sterm}$ (for multiple agents) which we call single term formulas in our notation. The formulas in $\mathcal{L}_\mathrm{sterm}$ are terms without any binary connectives and are inductively defined as 
\[ \phi := p \;| \neg \phi \;| K_i\phi \]
where $i\in I$ is  the set of agents, $p \in \mathsf{Prop} $ is the the set of
propositional letters and $K_i \phi$ is interpreted as agent $i$ knows $\phi$. We further classify single term formulas into simple single term and compound single term formulas to distinguish between single or multiple occurrences of an epistemic operator $E_i$ corresponding to an agent $i$. We give a few examples to motivate why we need separate analysis for the compound single term formulas. We have also considered the fragment $\mathcal{L}_\mathrm{mterm}$, where we allow for binary connectives. We present some preliminary results on characterization of conjunctions of single term formulas. We also have some general results which connect the class of successful formulas with other known classes of formulas, such as self refuting formulas \cite{holliday}. Our work is relevant as it provide insights to the nature and complexity of the problem in hand. The full characterization for the multi-agent case is still open and we briefly discuss possible ways to go about for solving the problem.

The organization of the paper is as follows: In Section 2 we state the
preliminaries and previous work. We present our characterization results on single term formulas $\mathcal{L}_\mathrm{sterm}$ in Section 3. In Section 4, we consider the multiple term language $\mathcal{L}_\mathrm{mterm}$ and present the characterization results for conjunctions of $\mathcal{L}_\mathrm{sterm}$ formulas. In Section 5, we present results which shed more light on the properties of
successful formulas and their connections with other known classes of formulas. We conclude the paper in Section 5, discussing a possible
approach to solve the general problem.
\section{Preliminaries and Previous Work}
In this section, we present the syntax and semantics of Public Announcement logic as given in \cite{DEL}. We also define successful formulas and list the existing results on characterization of successful formulas. 

The syntax of PAL ($\mathcal{L}_\mathrm{PAL}$) is given as follows:
\[\phi :=  \  p \;|\neg \phi\;|\; \phi \land \phi\; |\; \phi \lor \phi\; |\; K_i\phi\; |\; C\phi \;|\; [\phi]\;\psi \]
where $p\in \mathsf{Prop}$ is the set of propositional letters, $K_i \phi $ is interpreted as agent $i$ `knows' $\phi$ and $C\phi$ is interpreted as, it is common knowledge that $\phi$. We use the notation $[\phi]\psi$ for saying that $\psi$ is true after $\phi$ is announced.

An epistemic model is given by the triple $(W,R_i,V)$ where $W$ is the set of worlds and $R_i \subseteq W\times W$ is the accessibility relation  between the worlds for each agent $i \in I$, with $I$ being the index set of agents. The map $V: \mathsf{Prop}\rightarrow \mathcal{P}(W)$ is the valuation function specifying which propositional letters are true at a world $w\in W$. Since we restrict ourselves to the $S5$ case, we can assume that the relations $R_i$, for all agents $i\in I$ are equivalence relations. We use $R_I$ to denote the reflexive, transitive closure of union of all the relations $R_i$, $R_I = (\bigcup_{i\in I} R_i )^*$. Given a valuation $V$, we define the truth of a formula $\phi$ in a world $w$, denoted by $M, w\models \phi$ inductively below,
\begin{center}

\begin{tabular}{lcl}
$M, w \models  p$ & iff & $\;\;w \in V(p)$ \\
$M, w  \models \neg\phi $ & iff & $\;\;M, w\nvDash \phi$ \\
$M, w  \models \phi\land\psi$ & iff & $\;\;M, w  \models \phi$ and $M, w  \models \psi$ \\
$M, w  \models \phi\lor\psi$ & iff & $\;\;M, w  \models \phi$ or $M, w  \models \psi$ \\
$M, w  \models K_i \phi $ & iff & $\;\;\forall t \in W\;\mbox{s.t.}\; w R_i t \Rightarrow M, t  \models \phi$ \\
$M, w  \models C\phi$ & iff & $\;\;\forall t \in W \;\mbox{s.t.}\; w R_I t \Rightarrow M, t \models \phi$\\
$M, w  \models [\phi]\psi$ & iff & $\;\;M, w  \models \phi \Rightarrow M|_{\phi}, w  \models \psi$\\
\end{tabular}
\end{center}
where $M|_{\phi}= (W', R'_i, V')$ is the restriction of the model to the worlds where $\phi$ is true, and is defined as $W'=\{w\in W \;|\; M, w\models \phi\}$, $R'_i = R_i \cap (W'\times W')$ and $V'(p)=V(p)\cap W'$. 

We use $L_i\phi$ to denote the dual of $K_i\phi$, that is, $L_i\phi = \neg K_i \neg \phi $ and it is interpreted as agent $i$ considers it possible that $\phi$.

\begin{definition}[\textbf {Successful formulas}]
A formula $\phi$ is successful in PAL in {\bf S5} iff $[\phi]\phi$ is valid.
In other words, $M,w \models \phi$ implies $M|_\phi,w \models \phi$. 
\end{definition}
\begin{example}
The Moore sentence $p\land \neg Kp$ is a familiar example of an unsuccessful formula. We have the following model to illustrate why it is unsuccessful. Suppose we have two agents Ann and Bob. There is a butterfly on Bob's head but he doesn't know it, although Ann can see the butterfly. Let $p$ denote the sentence ``There is a butterfly on Bob's head" which is true at the actual world $w_2$. Since Bob does not know if there is a butterfly on his head, he cannot distinguish between the worlds $w_1$ and $w_2$. The models below represent the epistemic situation before and after the announcement.

\begin{figure}
\begin{center}
\tikzstyle{point}=[circle,draw]
\tikzstyle{c}=[circle,draw=red,minimum size=20pt,inner sep=0pt, ultra thick]
\tikzstyle{p}=[circle,draw,minimum size=10pt,inner sep=0pt]
\tikzstyle{ps}=[circle,draw,inner sep=1pt]
\tikzstyle{edge}=[style=-latex]
\begin{tikzpicture}[scale=0.5]
\path[rounded corners, draw=black] (-1.7,-0.6) rectangle (4.7,3.8);
  \node[ps] (n1) at (-0.5,1) [point] {$w_{1}$};
    \node[ps] (n4) at (3.5,1) [point] {$\underline{w_{2}}$} ;
  \node [below] at (n4.south) {$p$};
  %\node [above] at (n2.north) {$q$};
  \draw (n1) edge [edge,loop] node [above] {$R_a, R_{b}$} (n1);
  \draw (n1) edge node [above] {$R_{b}$} (n4);
  \draw (n4) edge [edge,loop] node [above] {$R_a, R_{b}$} (n4);
  \end{tikzpicture}
\hspace{15mm}
\begin{tikzpicture}[scale=0.5]
\path[rounded corners, draw=black] (-0.05,-0.6) rectangle (2.8,3.8);
 \node[ps] (n2) at (1.5,1) [point] {$\underline{w_{2}}$} ;
  \node [below] at (n2.south) {$p$};
  \draw (n2) edge [edge, loop] node [above] {$R_a, R_{b}$} (n2);
\end{tikzpicture}
\caption {Models before and after the announcement of $p \land \neg K_b p$}
\end{center}
\end{figure}
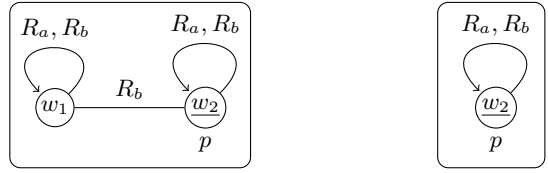
Before the announcement, $M, w_2\models p\land \neg K_b p$. After Ann makes the announcement, ``There is a butterfly on your head and you don't know it", the model changes to the one on the right, where Bob now knows that he has a butterfly on his head. The formula $p \land K_b p$ which is announced, is no longer true in the model on the right at $w_2$, and is therefore unsuccessful. 

\end{example}
The following result by van Benthem et al. \cite {ANJ,benthem} gives an immediate
subclass of formulas which are successful.
\begin{theorem}\label{succlass}
A formula is preserved under sub-models (of all relational models) iff it is
equivalent (in K) to a universal formula.
\end{theorem}
A universal formula in {\bf S5} is any formula which can be constructed
by $p,\land,\lor$ and $K$. This proves that the following sub-fragment of PAL is
successful, which we refer to as $\mathcal{L}_{suc}$
\[ \phi := p \;|\; \neg \phi \;|\; K_i \phi \;|\;\phi \land \phi \;| \;\phi \lor \phi \;|\; [\neg \phi] \phi \]
Other than this, \cite{DEL} also lists individual formulas, for instance $\neg
K_ap$, which are successful. The complexity of judging a formula to
be successful for multiple agents is shown to be PSPACE-complete in \cite{holliday}. 
In the same paper, a complete characterization for successful and
self-refuting formulas for S5 dynamic epistemic logic is also proposed for the single agent case. The
authors identify the source of all unsuccessful formulas as being a
\emph{Moorean-sentence}, and all self-refuting formulas as being a
\emph{Moore-sentence}. In addition, they define a \emph{super-successful
formula} as below.
\begin{definition}[\textbf {Super-successful formulas}]
A formula $\phi$ is super-successful iff given any M, for all $M^{'}$ such that
$M_\phi \subseteq M^{'} \subseteq M$  if $M,w \models \phi$ then $M^{'},w \models \phi$.
\end{definition}
It was shown in \cite{holliday} that super-successful formulas are closed under disjunction, but in general,
successful formulas are not, which is an important result.
\section{Characterization of $\mathcal{L}_\mathrm{sterm}$}
 We consider the subclass $\mathcal{L}_\mathrm{sterm}$ of \emph{single term formulas} in PAL.  The formulas in $\mathcal{L}_\mathrm{sterm}$ are inductively defined as 
\[ \phi := p \;| \neg \phi \;| K_i\phi \]
 Our choice of the subclass is appropriate in the
sense that it comprises of a basic language which can be used as a building
block for the complete $\mathcal{L}_\mathrm{PAL}$. 

 We use an operator variable $E_i$ to
stand for $K_i$ or $L_i$ in the description of formula forms. A formula involving operators with numeric subscripts has operators of only one type. For instance, $K_1\ldots K_n \phi $ denotes that there are exactly $n$, $K$ operators and no $L$ operator.
 We will work with formulas in negation normal form. We use $\alpha, \beta$ for denoting propositional formulas (without any epistemic operators) and $\phi, \psi$  are used to denote the formulas with epistemic operators.
\begin{definition}[\textbf{Single term formula}] A formula in negation normal form is
single term, if it is of the form $E_1 \ldots E_n\alpha$, where $E_i$ is
either $K_i$ or $L_i$ and $\alpha$ is a propositional formula.
\end{definition}

We now present the characterization results after the above mentioned notations. It is easy to see that any single term formula $E_1 \ldots E_n\alpha$ where $\alpha$ is a contradiction or a tautology is a successful formula.
\subsection{Simple single term formulas}
We first give a characterization for a simplified form of the single termed formulas.
\begin{definition}[\textbf{Simple single term formula}] A single term formula $E_1 \ldots E_n\alpha$ is said to be simple if for any agent $i\in I$, where $I$ is the index set for the set of agents, $E_i$ occurs at most once in $E_1 \ldots E_n.$ 
 \end{definition}
We further classify the simple single term formulas into $K$\emph {-simple single term} formulas which have only $K_i$ as the epistemic operators and \emph{L-simple single term} formulas which have only $L_i$ as the epistemic operators. The characterization of successful formulas is easy to see in both these cases and is given by the following proposition.
\begin{proposition}
\label{llsimple}
 K-simple single term and L-simple single term formulas are successful.
\end{proposition}
\begin{proof}
The $K$-simple single term formulas with only $K_i$ operators can be seen as a subclass of $\mathcal{L}_{suc}$ formulas, which we know are successful from Theorem \ref{succlass}. In case of $L$-simple single term formulas, suppose $M,w_1\models L_1 L_2\ldots L_n\alpha$. It gives us a chain of related worlds $w_1, \ldots, w_{n+1}$ such that $w_1 R_1 w_2 R_2 w_3\ldots w_n R_n w_{n+1}$ and $M,w_{n+1}\models\alpha$. We know that the frame is reflexive so $ L_1 L_2\ldots L_n\alpha$ is true at all the worlds in the chain, and no world is deleted after the public announcement of  $L_1 L_2\ldots L_n\alpha$. As a result, $M|_{ L_1 L_2\ldots L_n\alpha}, w_1\models  L_1 L_2\ldots L_n\alpha$ showing that $L$-simple single term formulas are successful.   
\end{proof}
We next define formulas which have both $L$ and $K$ epistemic operators and present characterization results for them. 
\begin{definition}[\textbf{$KL$-simple single term formula}] 
An KL-simple single term formula is a formula in which there exists at least one L operator in the scope of a K operator.
\end{definition}
\begin{example}As a simple example, the formula $K_1 L_2 L_3 K_4 L_5 \alpha$ is KL-simple single term, since the operator $L_2$ is in the scope of $K_1$. The formula $L_1 L_2 K_2 $ is not $KL$-simple single term, since there is no $L$ operator in the scope of the $K_2$ operator. We would like to stress the fact that the order of $K$ and $L$ operators in the formula does not make a difference as long as there is an $L$ operator in the scope of some $K$ operator.
\end{example}
\begin{proposition}\label{klsimple}
KL-simple single term formulas are unsuccessful.
\end{proposition}
\begin{proof}
In order to prove that $KL$-simple single termed formulas are unsuccessful, we first observe that for the simple case of 2 agents, the formula $K_1L_2\alpha$ is unsuccessful.  Consider the following counter-model (for the sake of clarity, we omit the reflexive arrows for each of the agents at every world ) 
\begin{figure}
\begin{center}
\tikzstyle{point}=[circle,draw]
\tikzstyle{c}=[circle,draw=red,minimum size=20pt,inner sep=0pt, ultra thick]
\tikzstyle{p}=[circle,draw,minimum size=10pt,inner sep=0pt]
\tikzstyle{ps}=[circle,draw,inner sep=1pt]
\tikzstyle{edge}=[style=-latex]
\begin{tikzpicture}[scale=0.5]
\path[rounded corners, draw=black] (-1.2,-0.1) rectangle (6.2,2.2);
  \node[ps] (n1) at (-0.5,1) [point] {$w_{1}$};
  \node[ps] (n2) at (1.5,1) [point] {$w_{2}$} ;
  \node[ps] (n3) at (3.5,1) [point] {$w_{3}$} ;
  \node[ps] (n4) at (5.5,1) [point] {$w_{4}$} ;
  \node [below] at (n1.south) {$p$};
  \node [below] at (n3.south) {$p$};
  %\node [above] at (n2.north) {$q$};
  \draw (n1) edge node [above] {$R_{1}$} (n2);
  \draw (n2) edge node [above] {$R_{2}$} (n3);
  \draw (n3) edge  node [above] {$R_{1}$} (n4);
  \end{tikzpicture}
\hspace{10mm}
\begin{tikzpicture}[scale=0.5]
\path[rounded corners, draw=black] (-1.2,-0.1) rectangle (2.3,2.2);
  \node[ps] (n1) at (-0.5,1) [point] {$w_{1}$};
  \node[ps] (n2) at (1.5,1) [point] {$w_{2}$} ;
  \node [below] at (n1.south) {$p$};
  \draw (n1) edge node [above] {$R_{1}$} (n2);
  \end{tikzpicture}
\caption {Models before and after the announcement of $K_1 L_2 p$}
\end{center}
\end{figure}
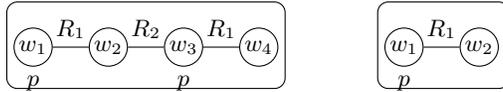

We need to show that $M, w_{1} \models K_1 L_2 p $ and $M |_{ K_1 L_2 p}, w_1 \models \neg K_1 L_2 p $. It can be easily checked that the formula $K_1 L_2 p$ is true at the worlds $w_1$ and $w_2$ but false at $w_3$ and $w_4$. Therefore after the public announcement of the formula $K_1 L_2 p $, the worlds where the formula is not true get deleted and we get the model on the right. In  the updated model $M |_{ K_1 L_2 p}$, the formula does not hold at $w_1$, that is, $M |_{ K_1 L_2 p}, w_1 \models \neg K_1 L_2 p $ which shows that $K_1 L_2 p$ is unsuccessful.

The above argument for the simple case easily generalizes to any $KL$-simple single term formula. Consider the  $KL$-simple single term formula $K_1\ldots L_2\ldots\alpha$ where $K_1$ is the first $K$ operator and $L_2$ is the first $L$ operator in the formula. We use the  notation $K_1 X L_2 Y \alpha$ for the formula $K_1\ldots L_2\ldots\alpha$, where $X$ and $Y$ are series of $K$ and $K, L$ epistemic operators respectively, in any arbitrary order. In order to show that $K_1\ldots L_2\ldots\alpha$ is unsuccessful, it suffices to use the same counter-model that we have above for the formula $K_1 L_2\alpha$. The reflexivity of the frame makes the formula $ K_1 X L_2 X\alpha $ true at the world $w_1$ and $w_2$ but false at $w_3$  and $w_4$, irrespective of the form of $X$ and $Y$. As a result after the announcement of the formula $K_1 X L_2 Y\alpha$, the model reduces to the one on the right where we have $M |_{ K_1 X L_2 Y \alpha}, w_1 \models \neg K_1X L_2 Y\alpha $, proving that $K_1 X L_2 Y\alpha$ is unsuccessful .
\end{proof}
\begin{definition}[\textbf{$LK$-simple single term formula}] 
An $LK$-simple single term formula is one which begins with $L$ operators and does not have any $L$ operator in the scope of a $K$ operator.
\end{definition}
\begin{example}
The simplest $LK$ simple single term formula would be $L_1 K_2 p$ with two epistemic operators. A more extensive example would be the formula $L_1\ldots L_m K_{m+1}\ldots K_n\alpha$. It is easy to observe that in general, any $LK$ simple single term formula will have  a series of $L$ operators followed by a series of $K$ operators, because of the restriction that we cannot have an $L$ operator in the scope of a $K$ operator. 
\end{example}
\begin{proposition}\label{lksucc}
$LK$-simple single term formulas are successful.
\end{proposition}
\begin{proof}
We first show that the $LK$-simple single term formula $L_1 K_2 \alpha$ is successful. Suppose, $M, w_1 \models L_1 K_2 p$ which implies $\exists \;w_2$, such that $w_1 R_1 w_2 $ and $M, w_2 \models K_2 p$. But since the frame is reflexive, we have $w_2 R_1 w_2$ and therefore  $M, w_2 \models L_1 K_2 p$. As a result, $w_2\in M|_{ L_1 K_2 p}$ which would make $ L_1 K_2 p$ true at $w_1$ in  $M|_{ L_1 K_2 p}$ since the relations are preserved under sub-models, thus proving that $ L_1 K_2 p$ is successful.  

   In order to prove that any $LK$-simple single term formula is successful, we use a similar argument as above.  Consider the formula  $L_1\ldots L_m K_{m+1}\ldots K_n\alpha$ which is true at a world $w_1$ in a model $M$. Since  $M, w_1 \models L_1\ldots L_m K_{m+1}\ldots K_n\alpha$, $\exists w_2$ such that $w_1 R_1 w_2$ and $M, w_2 \models L_2\ldots L_m K_{m+1}\ldots K_n\alpha$. We can repeat the same argument to get a chain of related worlds, $w_1R_1 w_2 R_2 w_3\ldots w_m R_m w_{m+1}$ such that $M, w_{m+1} \models K_{m+1}\ldots K_n\alpha$. Using reflexivity of $R_i$ for all $i\in I$ , we can show that all the worlds in the chain  $w_1R_1 w_2 R_2 w_3\ldots w_m R_m w_{m+1}$, will be present in the model after the announcement of the formula $L_1\ldots L_m K_{m+1}\ldots K_n\alpha$, since the formula is true in all the worlds connected to $w_1$ in the chain. Therefore, $M|_{ L_1\ldots L_m K_{m+1}\ldots K_n\alpha}, w_1\models L_1\ldots L_m K_{m+1}\ldots K_n\alpha$, proving that $LK$-simple single term formulas are successful.   
\end{proof}
%---------------------------------------------------------------------------------------------------------------------------------
\begin{comment}
\begin{enumerate}
\item
Formulas beginning with K operators
\begin{itemize}
\item
$K_1 \ldots K_n \alpha$: These are successful as they are in $L^{suc}$.
\item
$K_a \ldots L_b \ldots \alpha$: These are unsuccessful. They can be considered
as a generalization of the counter-model of $K_aL_b\alpha$.
\end{itemize}

\item
Formulas beginning with L operators
\begin{itemize}
\item
$L_1 \ldots L_n \alpha $: These are successful. This can be shown by reasoning
that each world involved from the base world to the last world belongs to
$M|_\phi$ using reflexivity of the relations.
\item
$L_1 \ldots L_m K_{1^{'}} \ldots K_{n^{'}} \alpha $: These are successful.By
using the proofs for `all K' and `all L' operators.
\item
$L_a \ldots K_b \ldots L_c \ldots \alpha$: These are unsuccessful. Proof is
same as for $K_a \ldots L_b \ldots \alpha $.
\end{itemize}
Note that these three cases cover all the formulas beginning with L operator.
\end{enumerate}
\end{comment}
%---------------------------------------------------------------------------------------------------------------------------------
A nice property of this class is that
all formulas which are successful are also super-successful. Thus, they will be
closed under disjunction. 
%--------------------------------------------------------------------------------------------------------------------------------------------------------------------------------------------------
\subsection{Compound single term formulas }
In this section we present characterization results for compound single term formulas where we allow multiple occurrences of  an epistemic operator $E_i$ within a formula.
\begin{definition}[\textbf{Compound single term formulas}] A single termed formula $E_1 \ldots E_n\alpha$ is said to be compound, if there is at least an agent $i\in I$, where $I$ is the index set for the set of agents, such that $E_i$ occurs more than once in the formula. 
 \end{definition}
We generalize the definition of $K$ and $L$-simple single term formulas to the compound case, as $K$ and $L$-compound simple term formulas, by allowing multiple occurrences of $E_i$ for agents  $i\in I$.
\begin{proposition}K-compound single term and L-compound single term formulas are successful.
\end{proposition}
\begin{proof}
The proof for the K-compound single term easily follows from the fact that they form a subclass of $\mathcal{L}_{suc}$ which are successful. For $L$-compound single term formulas, the proof is identical as in the case of simple formulas. Since the frame is reflexive, $M, w\models L_1L_1\alpha$ implies $M,w\models L_1\alpha$, so any multiple occurrences of epistemic operators occurring together can be reduced to a single occurrence. In case of multiple occurrence of epistemic operators not occurring together, we can use the same argument as in the proof of proposition \ref{llsimple}.
\end{proof}
The definition of the $KL$-simple single term formulas can be generalized to the setting of compound formulas by allowing multiple occurrences of epistemic operators corresponding to an agent. Unlike the simple formula case, where we have a single characterization result for all the $KL$-simple single term formulas, Proposition \ref{klsimple} does not hold for $KL$-compound simple term formulas.   
While we don't have a complete characterization of the $KL$ and $LK$-compound single term formulas, we present a few examples to show that some of the results for the simple formulas do not generalize to the compound formulas, which motivates separate and more general characterization results.  The following proposition shows that the formula $K_1L_2 K_1 p$ is successful, which would otherwise have been classified as unsuccessful in the simple single term case. 
\begin{proposition}\label{klsuccmp}
The compound single term formula $K_1 L_2 K_1 \alpha$ is successful.
\end{proposition}
 \begin{proof}
Suppose $M,w\models K_1 L_2 K_1 \alpha$ which implies that $\forall w_2$ such that $w_1 R_1 w_2, \; M, w_2\models L_2 K_1 p$.  We want to show that $M, w_2\models L_2 K_1 \alpha$. Consider an arbitrary $w'$ such that $w_2 R_1 w'$. Since $w_1 R_1 w_2$, by transitivity we have  $w_1 R_1 w'$ which makes $ L_2 K_1 \alpha$ true at $w'$, and therefore $M, w_2\models K_1L_2K_1 \alpha$. 

At $w_2$, we have $M, w_2\models L_2 K_1 \alpha$, which implies $\exists w_3$ s.t. $w_2R_2 w_3$ and $M, w_3\models K_1 \alpha$. Now using similar reasoning as above for $w_2$, using transitivity of $R_1$ we can show that $M, w_3\models K_1 L_2 K_1 \alpha$. Therefore, both $w_2$ and $w_3$ belong to the model $M|_{ K_1 L_2 K_1 \alpha}$, after the announcement of $K_1 L_2 K_1 \alpha$ which proves  $M|_{ K_1 L_2 K_1 \alpha}, w_1\models K_1L_2K_1\alpha$. 
\end{proof}
The generalization of above example to the case where we can have any number of epistemic operators and a characterization result for a sub-class of $KL$- compound formulas is quite involved and beyond the scope of this paper. Next, we have an example of the formula $K_1K_2 L_1 p$, which is unsuccessful as it would have been in the simple formula case, but the counter-model which we used earlier, doesn't suffice for this formula. This shows another deviation from the characterization in case of simple formulas. 
\begin{proposition}
The compound single term formula $K_1 K_2 L_1\alpha$ is unsuccessful. 
\end{proposition}
\begin{proof}
It is easy to check that the counter-model in Figure 1 does not work for the formula $K_1 K_2 L_1\alpha$, since it is true at all the worlds in the model and therefore no world is deleted from the model after the announcement of the formula. We extend the counter-model presented earlier so that it makes $K_1 K_2 L_1\alpha$ unsuccessful. 
\begin{figure}
\begin{center}
\tikzstyle{point}=[circle,draw]
\tikzstyle{c}=[circle,draw=red,minimum size=20pt,inner sep=0pt, ultra thick]
\tikzstyle{p}=[circle,draw,minimum size=10pt,inner sep=0pt]
\tikzstyle{ps}=[circle,draw,inner sep=1pt]
\tikzstyle{edge}=[style=-latex]
\begin{tikzpicture}[scale=0.5]
\path[rounded corners, draw=black] (-1.2,-0.3) rectangle (10.2,2.2);
  \node[ps] (n1) at (-0.5,1) [point] {$w_{1}$};
  \node[ps] (n2) at (1.5,1) [point] {$w_{2}$} ;
  \node[ps] (n3) at (3.5,1) [point] {$w_{3}$} ;
  \node[ps] (n4) at (5.5,1) [point] {$w_{4}$} ;
\node[ps] (n5) at (7.5,1) [point] {$w_{5}$} ;  
\node[ps] (n6) at (9.5,1) [point] {$w_{6}$} ;
\node [below] at (n1.south) {$\alpha$};
  \node [below] at (n4.south) {$\alpha$};
  %\node [above] at (n2.north) {$q$};
  \draw (n1) edge node [above] {$R_{1}$} (n2);
  \draw (n2) edge node [above] {$R_{2}$} (n3);
  \draw (n3) edge  node [above] {$R_{1}$} (n4);
  \draw (n4) edge  node [above] {$R_{1}$} (n5);
  \draw (n5) edge  node [above] {$R_{2}$} (n6);
 \end{tikzpicture}
\hspace{10mm}
\begin{tikzpicture}[scale=0.5]
\path[rounded corners, draw=black] (-1.2,-0.1) rectangle (2.3,2.2);
  \node[ps] (n1) at (-0.5,1) [point] {$w_{1}$};
  \node[ps] (n2) at (1.5,1) [point] {$w_{2}$} ;
  \node [below] at (n1.south) {$\alpha$};
  \draw (n1) edge node [above] {$R_{1}$} (n2);
  \end{tikzpicture}
\caption {Models before and after the announcement of $K_1 K_2 L_1 \alpha$}
\end{center}
\end{figure}
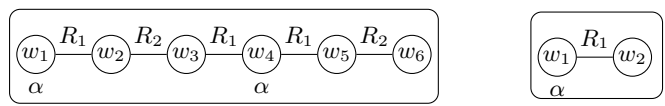
\end{proof}
We give an example of a formula $L_1K_2 K_3L_1\alpha$, beginning with an $L$ operator, which would have been classified as a $KL$-simple single term formula earlier and therefore unsuccessful, but in the compound case, it is successful. This further motivates the need for a separate characterization result for the compound case.
\begin{proposition}
The formula $L_1K_2 K_3L_1\alpha$ is successful.
\end{proposition}
\begin{proof}
The proof of this proposition is similar to proposition \ref{klsuccmp} and uses the idea that all the worlds make the formula true and are therefore contained in the sub-model.
Suppose $M,w_1\models L_1K_2 K_3L_1\alpha$, $\Rightarrow \exists w_2$, s.t $w_1R_1 w_2, M,w_2\models K_2 K_3L_1\alpha$. But the frame is reflexive, so  $ M,w_2\models L_1K_2 K_3L_1\alpha$, and therefore $w_2\in M|_{L_1K_2 K_3L_1\alpha}$. 

Now, $\forall w_3$ s.t. $w_2 R_2 w_3,  M, w_3\models K_3 L_1\alpha$. We can use the same argument as in the proof of proposition \ref{klsuccmp} to show that $M, w_3\models L_1K_2 K_3L_1\alpha$, so $w_3\in M|_{L_1K_2 K_3L_1\alpha}$. 

Since, $M, w_3\models L_1K_2 K_3L_1\alpha$, $\Rightarrow \exists w_4$, s.t $w_3R_1 w_4, M, w_4\models K_2 K_3L_1\alpha$. We leave it to the reader to check that, $M, w_3\models L_1K_2 K_3L_1\alpha$ iff $M, w_4\models L_1K_2 K_3L_1\alpha$. Once we have shown that, $w_4\in M|_{L_1K_2 K_3L_1\alpha}$ which implies $M|_{L_1K_2 K_3L_1\alpha}, w_1\models L_1K_2 K_3L_1\alpha$, thus proving $L_1K_2 K_3L_1\alpha$ is successful. 

\end{proof}

\section{Characterization of $\mathcal{L}_\mathrm{mterm}$}
In this section we present characterization results for the formulas in $\mathcal{L}_\mathrm{mterm}$ which includes Boolean combinations of simple single term formulas. We know from Section 3.1 that the $KL$-simple single term formulas are unsuccessful formulas. The proposition below generalizes the result to any number of Boolean conjunctions of unsuccessful formulas.

\begin{proposition} 
If $\phi$ and $\psi$ are unsuccessful simple single term formulas, their conjunction $\phi\land\psi$ is also unsuccessful.
\end{proposition} 
\begin{proof}
We showed in proposition \ref{klsimple} that $KL$-simple single term formulas are unsuccessful. It is easy to see that a conjunction of two simple single term $KL$ formulas will be unsuccessful. The counter-model for the conjunction will be the model which consists of the counter-models for each of the individual unsuccessful formulas sharing the real world as the common world. Since we are in the simple single term case, this counter-model is sufficient since $\phi$ and $\psi$ share no common epistemic operators and therefore the two counter-models corresponding to them will have no interaction. 
\end{proof}
The conjunction of a successful and an unsuccessful formula is unsuccessful as expected.
\begin{proposition}
If $\phi$ is successful and $\psi$ is unsuccessful, their conjunction $\phi\land\psi$ is unsuccessful.
\end{proposition}
\begin{proof}
In order to show that $\phi\land\psi$ is unsuccessful at a world $w\in W$, we use the counter-model starting at $w$ for proving that $\psi$ is unsuccessful and make $\phi$ true $w$. This is possible as long as the non epistemic parts of $\phi$ and $\psi$ don't depend on one another. 
\end{proof}
The analysis of conjunction of two successful formulas is more complicated and involves a number of cases and their success or failure depends on the non-epistemic parts of the formula. Recalling from Section 3.1, any simple single term successful formula is of the form $\alpha, \mathbf{K}\alpha, \mathbf{L}\alpha, \mathbf{LK}\alpha$, where $\mathbf{K}$ and $\mathbf{L}$ are series of $K$ and $L$ epistemic operators corresponding to different agents having at least 2 epistemic operators \footnote{for formulas with a single L and K, the characterization may differ for some cases, for eg. $\alpha\land L_1\mathbf{K}\beta$ is successful iff $\beta\rightarrow\alpha$, $L_1\alpha\land\L_2\beta$ is successful iff $\alpha\leftrightarrow\beta$ and $\alpha\land L_1\beta$ is successful iff $\alpha\rightarrow\beta$ or $\beta\rightarrow\alpha$}, and $\alpha$ is a propositional formula. 
\begin{proposition} The conjunction of simple single term successful formulas is successful or unsuccessful subject to the following conditions
\begin{enumerate}
\item
The conjunctions $\alpha\land \beta$, $\alpha\land \mathbf{K}\beta$ and $\mathbf{K}\alpha\land \mathbf{K'}\beta$ are successful where, $\alpha$ and $\beta$ are formulas without epistemic operators and $\mathbf{K}$ and $\mathbf{K'}$ are series of $K$ epistemic operators.
\item
The conjunction $\alpha\land\mathbf{L}\beta$ is successful iff $\alpha\rightarrow\beta$. 
\item
The conjunction $\alpha\land\mathbf{LK}\beta$ is unsuccessful. 
\item
The conjunction $\mathbf{K}\alpha\land\mathbf{L}\beta$ is successful iff $\alpha\rightarrow\beta$. 
\item
The conjunctions $\mathbf{K}\alpha\land\mathbf{LK'}\beta$, $\mathbf{L}\alpha\land\mathbf{L'}\beta$, $\mathbf{L}\alpha\land\mathbf{L'K}\beta$ and $\mathbf{LK}\alpha\land\mathbf{L'K'}\beta$ are unsuccessful.
\end{enumerate}
\end{proposition}
\begin{proof}
The proof for 1 is trivial since we know from Theorem \ref{succlass} that universal formulas are preserved under sub-models. 

For 2, if we assume $\alpha\rightarrow\beta$, then proving $\alpha\land \mathbf{L}\beta$ is successful is easy. To see why assume $M, w \models \alpha\land \mathbf{L}\beta$. Since, $\alpha$ is a propositional formula it will be preserved in every sub-model  $M'\subseteq M$  and therefore in particular, $M|_{\alpha\land \mathbf{L}\beta},w\models \alpha$. But, $\alpha\rightarrow \beta$ and the frame is reflexive so we have $M|_{\alpha\land \mathbf{L}\beta},w\models \alpha\land \mathbf{L}\beta$. The converse direction can be proved using a contrapositive argument. Suppose $\alpha\nrightarrow\beta$, we can construct a counter-model to show $\alpha\land \mathbf{L}\beta$, where $\mathbf{L}=L_1 L_2\ldots L_n$, is unsuccessful in the following way.  Let $w_1R_1w_2R_2\ldots w_n R_n w_{n+1}$ be a set of related worlds. We make $\alpha$ true at only $w_1$ and false at all other worlds and $\beta$ true only at $w_{n+1}$. After the announcement of $\alpha\land \mathbf{L}\beta$, only the world $w_1$ will remain in the sub-model making the formula $\alpha\land\, \mathbf{L}\beta$ false at $w_1$.

For 3, we can construct a counter-model in the same way as above. Assume, $M_1\models \alpha \land \,\mathbf{L}\mathbf{K}\beta$. We have $\mathbf{LK}=L_1\ldots L_m K_{m+1}\ldots K_n$.  So, there exists a chain of related worlds, $w_1R_1 w_2 R_2\ldots w_n R_n w_{n+1}$ such that $M, w_{n+1}\models K_{m+1}\ldots K_n\alpha$. In order to have a counter-model, we make $\alpha$ true only at $w_1$ and false at all other worlds and $\beta$ false at a $w'$ related to $w$. One can check that irrespective of $\alpha\leftrightarrow\beta$, $ \alpha \land \,\mathbf{L}\mathbf{K}\beta$ is unsuccessful.   

We leave the proof of 4, which is similar to 2, and of 5 which is similar to 2 to the reader.
\end{proof}

%---------------------------------------Other results----------------------------------------------------------------------------
\section{Other results}
The following section is a mixed bag of auxiliary results relating to successful
formulas which might come in handy for further analysis of different classes.
The following theorem relates to the class of \emph{successful} and
\emph{super-successful} formulas.
\begin{theorem}
The following class of S5- PAL formulas are truth-preserved under super-models:
\[ \phi := p\;|\;\neg p\;|\;\phi \land \psi \;|\;\phi \lor \psi \;|\;L_a\phi \;|\;\neg [
\phi ] \neg \psi\]
\end{theorem}
\begin{proof}
We prove the above result by induction on the complexity of the formula $\phi$. Consider the case, $\phi = p $ and assume $M,w\models \phi$ . We know that the truth of a propositional formula is local, that is, depends only on the current state, so any super-model of $M$ will contain $w$ and hence the statement is true for any propositional formula.\\
If $M,w \models \phi \land \psi$ then $M,w \models \phi$ and $M,w \models
\psi$. By induction hypothesis, both $\phi$ and $\psi$ are true in any super-model $M'$ of $M$.
Hence, $M^{'},w \models \phi \mbox{ and } M^{'},w \models \psi$ which shows $M^{'},w \models
\phi \land \psi$. The proof is similar as above for $\phi \lor \psi$.\\
If $M,w \models L_a\phi$, then $M,w' \models \phi$ where $w R_a w'$. Thus
$M',w' \models \phi$ for a super-model $M'$ of $M$ by
induction hypothesis which shows $M^{'},w \models L_a\phi$.\\
If $M,w \models \neg[\phi] \neg \psi$, then $M,w \models \phi$ and $M|_\phi,w
\not\models \neg \psi$ holds, i.e. $M|_\phi,w \models
\psi$.
Now, by induction hypothesis, $M',w \models \phi$.  Consider, $M|_\phi$, $M^{'}|_\phi$, and a world $s \in M|_\phi$. Assume $M,s \models \phi$,  which gives us $M',s \models \phi$ by induction hypothesis, and therefore $s \in M'|_\phi$. This shows that $M|_\phi \subseteq M'|_\phi$, that is,  $M'|_\phi$ is a super-model of $M|_\phi$. Hence, by induction hypothesis, $M'|_\phi,w \models \psi$ which finally proves $M',w \models \neg [ \phi ] \neg \psi.$
\end{proof}
The above result implies that if any successful formula belongs to this class, it must be
super-successful, as $M,w \models \phi \Rightarrow M|_\phi,w \models \phi$ and by the
above formula any $M'$ such that $ M|_\phi \subseteq M'$, and $ M',w \models\phi$.

\noindent We have seen that the class of \emph{self-refuting formulas} is
another class of formulas other than \emph{successful formulas} which are
interesting. 
\begin{definition}[\textbf{Self-refuting formulas}]
A formula is self-refuting iff $[\phi]\neg\phi$ is valid.
\end{definition}
The following theorem links the two classes of formulas.
\begin{theorem}
A formula in S5-PAL is a contradiction iff it is both successful and
self-refuting.
\end{theorem}
\begin{proof}
It trivially follows from the definition, that a contradiction is both successful and self refuting. For the converse, suppose $\phi$
is both successful and self-refuting. Then $[\phi]\phi$ and $[\phi]\neg\phi$
are valid. Suppose for a given pointed model $(M,w)$, if we have $M,w \models \phi$ then
$M|_\phi,w \models \phi$ and $M|_\phi,w \models \phi$, which is a contradiction to our initial assumption. Therefore, $M,w \nvDash\phi$ which shows $\phi$ is a contradiction.
\end{proof}

In \cite{holliday}, it has been shown that successful formulas are not closed under
disjunction for the single agent case. We have a result along similar lines for the closure under $L$ operator.
\begin{theorem}
The class of successful formulas is not preserved under L operator in the multi-agent case.
\end{theorem}
\begin{proof}
 $L_1K_aK_bL_1 p$ is successful, while $L_2L_1K_aK_bL_1
p$ is not, the proof of which is given by the counter-model in the appendix.
\end{proof}
\section{Discussion for the General Case and Conclusion}
One can see that the work we have presented in this paper opens up new directions to be explored. We list a few questions answering which, may help to give a complete characterization. We have seen in Section 3.2 that the characterization for the compound single term formulas is quite involved and does not follow as a generalization of the simple single term formulas. We have some preliminary results regarding their characterization which we have not presented in this paper. The idea is to have additional conditions on $KL$ and $LK$ simple single term formulas which allows us to have their complete characterization. We don't have any results on the compound multiple term formulas involving boolean connectives, which would be interesting to look into.

We have seen that the multiple agent scenario is complicated even for single
terms as opposed to single agent case, where single terms are always successful.
Recursively combining the single terms using conjunction or disjunction and then binding the
whole formula within an epistemic operator may result in formulas of increasing
complexity. The way out may be finding a ``normal form" in which the formula can be expressed in an equivalent conjunctive normal form (c.n.f.) or disjunctive normal form (d.n.f.). Alternatively, as a weaker attempt, we may be able to find a class of formulas which in spite of not being logically equivalent, can only be successful iff the original formula is successful. We believe that such a reduction algorithm would be of great help in avoiding the complex cases arising out of Boolean combinations of formulas.

\noindent In a nutshell, a possible way of approaching the task of syntactic
characterization could be:
\begin{enumerate}
\item
Finding a normal form of the formulas which preferably are in c.n.f or d.n.f of
single-term formulas
\item
Propose a way to classify the formulas thus obtained from 1.
\end{enumerate}
Our classification above proceeds in direction of achieving 2. Combining the
ideas and results, and those in \cite{holliday} for single-agent
classification, we might be able to achieve 2. But whether 1 holds or not is
something which is unknown to us at this stage and may be very important with respect to the
difficulty of solving the problem of characterizing successful formulas in PAL.
\bibliographystyle{abbrv}
\bibliography{esslli}

\begin{thebibliography}{10}

\bibitem{ANJ}
H.~Andr\'{e}ka, I.~N\'{e}meti, and J.~van Benthem.
\newblock Modal languages and bounded fragments of predicate logic.
\newblock {\em Journal of Philosophical Logic}, 27:217--224, 1998.

\bibitem{BaltagvDM07}
A.~Baltag, H.~{van Ditmarsch}, and L.~Moss.
\newblock Epistemic logic and information update.
\newblock In P.~Adriaans and J.~{van Benthem}, editors, {\em Handbook on the
  Philosophy of Information}. Elsevier, 2008.

\bibitem{BaltagMS98}
A.~Batlag, L.~S. Moss, and S.~Solecki.
\newblock The logic of public announcements and common knowledge and private
  suspicions.
\newblock In I.~Gilboa, editor, {\em TARK}, pages 43--56. Morgan Kaufmann,
  1998.

\bibitem{gebrandy}
J.~Gebrandy.
\newblock {\em {Bisimulations on Planet Kripke}}.
\newblock PhD thesis, University of Amsterdam, 1998.

\bibitem{Hintikka}
J.~Hintikka.
\newblock {\em Knowledge and Belief: An Introduction to the Logic of the Two
  Notions}.
\newblock Cornell University Press, Ithaca, N. Y., 1962.

\bibitem{holliday}
W.~H. Holliday and T.~F.~I. III.
\newblock Moorean phenomena in epistemic logic.
\newblock In {\em Advances in Modal Logic'10}, pages 178--199, 2010.

\bibitem{moore}
G.~E. Moore.
\newblock A reply to my critics.
\newblock {\em The Philosophy of G.E. Moore, The Library of Living
  Philosophers}, 4:535--677, 1942.

\bibitem{plaza}
J.~Plaza.
\newblock Logics of public communications.
\newblock {\em Synthese}, 158:165--179, 2007.

\bibitem{benthemone}
J.~van Benthem.
\newblock One is a lonely number : On the logic of communication.
\newblock {\em Logic Colloquium 02}, (December):1--37, 2002.

\bibitem{benthemld}
J.~van Benthem.
\newblock Open problems in logical dynamics.
\newblock In D.~M. Gabbay, S.~S. Goncharov, and M.~Zakharyaschev, editors, {\em
  Mathematical Problems from Applied Logic I}, volume~4 of {\em International
  Mathematical Series}, pages 137--192. Springer New York, 2006.

\bibitem{BenthemOpen}
J.~{van Benthem}.
\newblock Open problems in logical dynamics.
\newblock In D.~Gabbay, S.~Goncharov, and M.~Zakharyashev, editors, {\em
  Mathematical Problems from Applied Logic I}, pages 137--192. Springer, 2006.

\bibitem{ditmarschphd}
H.~van Ditmarsch.
\newblock {\em {Knowledge games}}.
\newblock PhD thesis, University of Groningen, 2000.

\bibitem{DEL}
H.~van Ditmarsch, W.~van~der Hoek, and B.~Kooi.
\newblock {\em Dynamic Epistemic Logic}.
\newblock Springer, 2008.

\bibitem{secret}
H.~P. van Ditmarsch and B.~Kooi.
\newblock The secret of my success.
\newblock {\em Synthese}, 2006:2006, 2004.

\bibitem{benthem}
A.~Visser, J.~van Benthem, D.~de~Jongh, and G.~R.~R. de~Lavalette.
\newblock Nnil, a study in intuitionistic propositional logic.
\newblock {\em Logic Group Preprint Series}, 111:535--677, 1994.

\end{thebibliography}
\section{Appendix}
In the counter-model below, $p$ is true only the worlds $x$, $y$ and $z$ and
false in rest of the worlds. Clearly $M, w \models L_2 L_1 K_a K_b L_1 p$. Also,
$ L_2 L_1 K_a K_b L_1 p$ is true in the worlds $s, t, u, y, z$ and $v$ and false in the world $x, x'$ and $v$ . Note
that if $v$ was combined to $s$ by a $1$-edge, then both $v$ and $x$ would have satisfied $L_1 K_a K_b L_1 p$. Hence, this construction cannot be used as a counterexample of $L_1 K_a K_b L_1 p$ (which in fact is successful). Thus, the restricted model has only $v$ and not $x, x'$ and $v'$ . So, in the restricted model $M |_{\phi}, t \models K_a K_b p$. Thus, $M|_\phi , w \models L_2 L_1 K_a K_b L_1 p$. Thus $\phi = L_2 L_1 K_a K_b L_1 p$ is not successful.

\begin{figure}\label{appctrmodel}
\begin {center}
\includegraphics[scale=0.75]{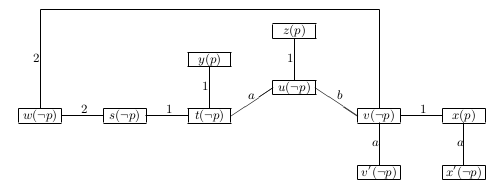}
\end{center}
\caption{Counter-model for $L_2 L_1 K_a K_b L_1 p$}
\end{figure}

\end{document}